\documentclass[letterpaper, 11pt]{article}

\newenvironment{proof}{{\bf Proof:}}{\hspace*{\fill}\(\Box\)}
\newenvironment{claimproof}{{\bf Proof of the Claim:}}{\hspace*{\fill}\(\star\)}

\title{Byzantine Vector Consensus in Complete Graphs~\thanks{This research is supported
in part by
 National
Science Foundation awards CNS-1059540 and CNS-1115808 and the Cullen Trust for Higher Education. Any opinions, findings, and conclusions or recommendations expressed here are those of the authors and do not
necessarily reflect the views of the funding agencies or the U.S. government.}
}

\author{{\normalsize\bf Nitin H. Vaidya}\\ \normalsize University of Illinois at Urbana-Champaign\\
\normalsize nhv@illinois.edu\\  \normalsize Phone: +1 217-265-5414\\ ~\\
 \normalsize {\bf Vijay K. Garg}\\ \normalsize University of Texas at Austin\\ \normalsize garg@ece.utexas.edu\\ \normalsize Phone: +1 512-471-9424}

\date{February 11, 2013}

\usepackage{graphicx}

\newcommand{\comment}[1]{}

\usepackage{graphicx}
\usepackage{latexsym}

\newcommand{\nchoosek}[2]{{#1 \choose #2}}

\newcommand{\scriptw}{\mathcal{W}}
\newcommand{\scripta}{\mathcal{A}}

\newcommand{\scripte}{\mathcal{E}}

\newcommand{\scripth}{\mathcal{H}}
\newcommand{\scriptp}{\mathcal{P}}

\newtheorem{theorem}{Theorem}

\newtheorem{claim}{Claim}

\newtheorem{definition}{Definition}

\newtheorem{lemma}{Lemma}

\newcommand{\fillbox}{\hspace*{\fill}\(\Box\)}

\def\noflash#1{\setbox0=\hbox{#1}\hbox to 1\wd0{\hfill}}

\newcommand{\bfA}{{\bf A}}

\newcommand{\bfM}{{\bf M}}

\newcommand{\bfT}{{\bf T}}
\newcommand{\bfx}{{\bf x}}
\newcommand{\bfv}{{\bf v}}
\newcommand{\bfw}{{\bf w}}
\newcommand{\bfr}{{\bf r}}
\newcommand{\bfs}{{\bf s}}
\newcommand{\bfz}{{\bf z}}

\newcommand{\I}{\Gamma(S)}

\begin{document}

\maketitle

\thispagestyle{empty}

\begin{abstract}
{\normalsize
\setlength {\parskip}{8pt}                                             
Consider a network of $n$ processes each of which has
a $d$-dimensional vector of reals as its {\em input}. 
Each process can communicate directly with all the processes in the system;
thus the communication network is a {\em complete graph}.
All the communication channels are reliable and FIFO (first-in-first-out).
The problem of {\em Byzantine vector consensus} (BVC) requires agreement on a
$d$-dimensional vector that is in the {\em convex hull} of the $d$-dimensional
input vectors at the non-faulty processes.
We obtain the following results for Byzantine vector consensus in {\em complete graphs} while
tolerating up to $f$ Byzantine failures:
\begin{itemize}
\item 
We prove that in a synchronous system, $n\geq \max(~3f+1,~(d+1)f+1~)$
is necessary and sufficient for achieving Byzantine vector consensus.

\item In an asynchronous system, it is known that {\em exact}
consensus is impossible in presence of faulty processes. For an 
asynchronous system, we prove that
$n\geq (d+2)f+1$ is necessary and
sufficient to achieve {\em approximate} Byzantine vector consensus.
\end{itemize}
Our sufficiency proofs are constructive. We show sufficiency by providing
explicit algorithms that solve exact BVC in synchronous
systems, and approximate BVC in asynchronous systems.

We also obtain tight bounds on the number of processes for achieving
BVC using algorithms that
are restricted to a simpler communication pattern.

\setlength {\parskip}{5pt}                                             

}
\end{abstract}


\newpage

\setcounter{page}{1}

\section{Introduction}
\label{s_intro}

This paper addresses {\em Byzantine vector consensus} (BVC), wherein the input at each process is a
$d$-dimensional vector of reals, and each process is expected to decide on a
{\em decision vector} that is in the
{\em convex hull} of the input vectors at the non-faulty processes. 
The system consists of $n$ processes in $\scriptp=\{p_1,p_2,\cdots,p_n\}$. 
We assume $n>1$, since consensus is trivial for $n=1$.
At most $f$ processes may be Byzantine faulty, and may behave arbitrarily
\cite{lamport82}.
All processes can communicate with each other directly on
{\em reliable} {\em FIFO} (first-in first-out) channels. Thus,
 the communication network is a {\em complete graph}.
The input {\em vector} at each process may also be viewed as a {\em point} in the $d$-dimensional Euclidean space
${\bf R}^d$, where $d>0$ is a finite integer.
Due to this correspondence, we use the terms {\em point} and {\em vector}
interchangeably. Similarly, we interchangeably refer to the $d$ {\em elements}
of a vector as {\em coordinates}.
We consider two versions of the Byzantine vector consensus (BVC) problem, {\em
Exact BVC}\, and {\em Approximate BVC}.

\paragraph{\bf Exact BVC:}
Exact Byzantine vector consensus must satisfy the following three conditions.
\begin{itemize}
\item {\em Agreement}: The decision (or output) vector at all the non-faulty processes must be identical.
\item {\em Validity}: The decision vector at each non-faulty process must be in the convex hull of the input vectors at the non-faulty processes.
\item {\em Termination}: Each non-faulty process must terminate after a finite amount of time.
\end{itemize}

The traditional consensus problem \cite{lynch:book96,garay1998} is obtained when $d=1$; we refer to this as {\em scalar}
consensus.
$n\geq 3f+1$ is known to be necessary and sufficient 
for achieving Byzantine {\em scalar} consensus in complete graphs \cite{lamport82,lynch:book96}.  
We observe that simply performing {\em scalar} consensus 
on each dimension of the input vectors independently does not solve
the {\em vector} consensus problem. In particular, even if validity
condition for {\em scalar consensus} is satisfied for each dimension of the
vector separately, the above {\em validity} condition of vector consensus may not necessarily be satisfied.
 For instance, suppose that there are four processes, with one faulty process.
 Processes $p_1,p_2$ and $p_3$ are non-faulty, and have the following
3-dimensional input vectors, respectively:
$\bfx_1 = [\frac{2}{3}, \frac{1}{6}, \frac{1}{6}] $,
 $\bfx_2 = [\frac{1}{6}, \frac{2}{3}, \frac{1}{6}] $,
$\bfx_3 = [\frac{1}{6}, \frac{1}{6}, \frac{2}{3}] $.
Process $p_4$ is faulty.
If we perform Byzantine {\em scalar} consensus on each dimension of the vector separately, then the
processes may possibly agree on the decision vector 
$[\frac{1}{6}, \frac{1}{6}, \frac{1}{6}]$,
each element of which satisfies {\em scalar} validity condition {\em along each dimension}
separately; however, this decision vector {\em does not} satisfy the validity condition for BVC
because it is {\em not} in the convex hull of input vectors of non-faulty processes.
In this example, since every non-faulty process has a probability vector as its input vector,
BVC validity condition requires that the decision vector should also be a probability vector.
In general, for many optimization problems \cite{boyd2004},
 the set of feasible solutions is
a convex set in Euclidean space.
Assuming that every non-faulty process proposes a feasible solution,
BVC guarantees that the vector decided is also a feasible solution. Using scalar consensus along
each dimension is not sufficient to provide this guarantee.

\paragraph{\bf Approximate BVC:}
In an {\em asynchronous} system, processes may take steps at arbitrary relative speeds, and there
is no fixed upper bound on message delays.
Fischer, Lynch and Paterson \cite{lynch:imposs:async} proved that exact consensus is impossible
in asynchronous systems in the presence of even a single crash failure.
 As a way to circumvent
this impossibility result, Dolev et al. \cite{AA_Dolev_1986} introduced the notion of
{\em approximate} consensus, and proved the correctness of an algorithm
for approximate Byzantine {\em scalar} consensus in asynchronous systems when $n\geq 5f+1$.
Subsequently, Abraham, Amit and Dolev \cite{Abraham_optimalresilience04} established that
approximate Byzantine {\em scalar} consensus is possible in asynchronous systems if $n\geq 3f+1$.
Other algorithms for approximate consensus have also been proposed (e.g., \cite{gradecast,AA_Fekete_aoptimal}).
We extend the notion of approximate consensus to {\em vector} consensus.
{\em Approximate BVC} must satisfy the following conditions:

\begin{itemize}
\item {\em $\epsilon$-Agreement}: For $1\leq l\leq d$,
the $l$-th elements of the decision vectors at any two non-faulty processes
must be within $\epsilon$ of each other, where $\epsilon>0$ is a
pre-defined constant.

\item {\em Validity}: The decision vector at each non-faulty process must be in the convex hull of the input vectors at the non-faulty processes.

\item {\em Termination}: Each non-faulty process must terminate after a finite amount of time.
\end{itemize}
%
The main contribution of this paper is to establish the following bounds for {\em complete graphs}.

\begin{itemize}
\item 
In a synchronous system,
$n\geq \max(3f+1,(d+1)f+1)$ is necessary and sufficient for {\em Exact BVC}
in presence of up to $f$ Byzantine faulty processes.
~~(Theorems \ref{t_exact_nec} and \ref{t_exact_suff}).

\item
In an asynchronous system,
$n\geq (d+2)f+1$ is necessary and sufficient for {\em Approximate BVC}
in presence of up to $f$ Byzantine faulty processes.
~~(Theorems \ref{t_approx_nec} and \ref{t_approx_suff}).
\end{itemize}

Observe that the two bounds above are different when $d>1$, unlike the case of
$d=1$ (i.e., scalar consensus). When $d=1$, in a complete graph, $3f+1$ processes are sufficient
for exact consensus in synchronous systems, as well as approximate
consensus in asynchronous systems \cite{Abraham_optimalresilience04}. For $d>1$, the lower bound for
asynchronous systems is larger by $f$ compared to the bound for synchronous
systems. 


In prior literature, the term {\em vector consensus} has also been used to
refer to another form of consensus, wherein the input
at each process is a {\em scalar}, but the agreement is on a vector
containing these scalars \cite{doudou98,neves05}. Thus, our results 
are for a different notion of consensus.

\subsubsection*{\normalsize Simpler (Restricted) Algorithm Structure}

In prior literature, iterative algorithms with very simple structure
have been proposed to achieve {\em approximate} consensus, including
asynchronous approximate Byzantine scalar consensus \cite{AA_Dolev_1986}
in {\em complete} graphs, and synchronous as well as asynchronous
approximate Byzantine consensus in
{\em incomplete} graphs \cite{vaidya12podc}.
Section \ref{s_simple}
extends these simple structures
to vector consensus in complete graphs,
and obtains the following tight bounds:
(i) $n\geq (d+2)f+1$ for synchronous systems,
and
(ii) $n\geq (d+4)f+1$ for asynchronous systems.
Observe that the bound for the simple iterative algorithms in asynchronous
systems is larger by $2f$ when compared to the bound stated earlier: this
is the cost of restricting the algorithm structure.
This $2f$ gap is analogous to that between the sufficient condition of $n\geq 3f+1$ for
asynchronous scalar consensus proved by Abraham, Amit and Dolev \cite{Abraham_optimalresilience04},
the sufficient condition of $n\geq 5f+1$ demonstrated by Dolev et al. \cite{AA_Dolev_1986}
using a simpler algorithm.

\subsubsection*{Our Notations}
Many notations introduced throughout the paper are also summarized in
Appendix \ref{a_notations}.
We use operator $|\,.\,|$ to obtain the size of 
a {\em multiset} or a {\em set}. We use operator $\parallel .\parallel$
to obtain the absolute value of a scalar.




\section{Synchronous Systems}
\label{s_sync}

In this section, we derive necessary and sufficient conditions for exact BVC
in a synchronous system with up to $f$ faulty processes.
The discussion in the rest of this paper assumes that the network is a {\em complete graph},
even if this is not stated explicitly in all the results.

\subsection{Necessary Condition for Exact BVC}
\label{ss_sync_nec}

\begin{theorem}
\label{t_exact_nec}
$n\geq \max(3f+1,(d+1)f+1)$ is necessary for Exact BVC in a synchronous system.
\end{theorem}
\begin{proof}
From \cite{lamport82,lynch:book96},
we know that, for $d=1$ (i.e., scalar consensus),
$n\geq 3f+1$ is a necessary condition for achieving exact Byzantine consensus
in presence of up to $f$ faults.
If we were to restrict the $d$-dimensional input vectors to have identical $d$ elements,
then the problem of vector consensus reduces to scalar consensus. Therefore, $n\geq 3f+1$
is also a necessary condition for {\em Exact BVC} for arbitrary $d$.
Now we prove that $n\geq (d+1)f+1$ is also a necessary condition.

First consider the case when $f=1$, i.e., at most one process may be faulty.
Since none of the non-faulty processes know which process, if any, is faulty,
as elaborated in Appendix \ref{a_sync}, the 
decision vector must be in the convex hull of each multiset containing
the input vectors of $n-1$ of the processes (there are $n$ such multisets).\footnote{Since
 the state of two processes may be identical, we use a {\em multiset}
 to represent the collection of the states of a subset of processes.
Appendix \ref{a_multisets} elaborates on the notion of multisets.}
Thus, this intersection must be non-empty,
for all possible input vectors at the $n$ processes.
(Appendix \ref{a_sync} provides further clarification.)
We now show that the intersection may be empty when
$n=d+1$; thus, $n=d+1$ is not sufficient for $f=1$.

Suppose that $n=d+1$.
Consider the following set of input vectors.
The input vector of process $p_i$, where $1\leq i\leq d$, is
a vector whose $i$-th element is $1$, and the remaining elements are 0.
The input vector at process $p_{d+1}$ is the all-0 vector (i.e., the vector
with all elements 0). Note that the $d$ input vectors at $p_1,\cdots,p_d$
form 
the standard basis for the $d$-dimensional vector space.
Also, none of the $d+1$ input vectors can be represented as a convex
combination of the remaining $d$ input vectors.
For $1\leq i\leq d+1$,
let $Q_i$ denote the convex hull of the inputs at the $n-1=d$
processes in $\scriptp-\{p_i\}$. 
We now argue that 
$\cap_{i=1}^{d+1}\, Q_i$ is empty.

For $1\leq i\leq d$, observe that for all the points in $Q_i$, 
the $i$-th coordinate is 0. Thus, any point
that belongs to the intersection $\cap_{i=1}^d\, Q_i$ must have all its
coordinates 0. That is, only the all-0 vector belongs to
$\cap_{i=1}^d\, Q_i$.
Now consider $Q_{d+1}$, which is the convex hull of
the inputs at the first $d$ processes. Due to the choice of the inputs
at the first $d$ processes, the origin (i.e., the all-0 vector) does not belong
to $Q_{d+1}$. From the earlier observation on $\cap_{i=1}^d\, Q_i$, it then
follows that $\cap_{i=1}^{d+1}\,Q_i=\emptyset$.
Therefore, the {\em Exact BVC} problem for $f=1$ cannot be solved with $n=d+1$.
Thus, $n=d+1$ is not sufficient.
It should be easy to see that $n\leq d+1$ is also not sufficient.
Thus, $n\geq d+2$ is a necessary condition for $f=1$.

Now consider the case of $f>1$.
Using the commonly used simulation approach \cite{lamport82},
we can prove that $(d+1)f$ processes are not sufficient.
In this approach, $f$ {\em simulated processes} are implemented
by a single process. If a correct algorithm were to exist for tolerating
$f$ faults among $(d+1)f$ processes, then we can obtain a correct
algorithm to tolerate a single failure among $d+1$ processes,
contradicting our result above.
\comment{+++++++++++++++++++++++++++
The simulation approach
simulates $f$ processes using a 
Suppose that there exists an algorithm $\scripta$ that can solve
the Exact BVC problem with $f$ faults in a system with $n\leq (d+1)f$. 
We will use $\scripta$ to solve Exact BVC in
a system consisting of $d+1$
processes, of which at most 1 process may be faulty. Suppose that each of the $d+1$ processes
simulates $f$ pseudo-processes. This results in $(d+1)f$ pseudo-processes. Of these,
at most $f$ pseudo-processes may be faulty, because at most 1 of the processes may be faulty.
The pseudo-processes simulated by a certain process have the same input as that process.
The output of a process is obtained as the output of one of the pseudo-processes simulated by
that process.
Since algorithm $\scripta$ can achieve consensus among the $(d+1)f$ pseudo-processes
with up to $f$ failures, $\scripta$ can also be used to achieve {\em Exact BVC} among
$d+1$ processes with 1 failure. This contradicts the result above for $f=1$.
Thus, $(d+1)f$ processes are insufficient to tolerate $f$ faults, and
$n\geq (d+1)f+1$ is a necessary condition.
++++++++++++++++++++}
Thus, $n\geq (d+1)f+1$ is necessary for $f\geq 1$.
(For $f=0$, the necessary condition holds trivially.)
\end{proof}

\subsection{Sufficient Condition for Exact BVC}
\label{ss_sync_suff}

We now present an algorithm for Exact BVC in a synchronous system, and prove its correctness
in a complete graph with $n\geq \max(3f+1,(d+1)f+1)$.
The algorithm uses function $\Gamma(Y)$ defined below, where $Y$ is a multiset of points.
$\scripth(T)$ denotes the convex hull of a multiset $T$.
\begin{eqnarray} \Gamma(Y) = \cap_{T\subseteq Y,\, |T|=|Y|-f}~\scripth(T).
\label{e_I}
\end{eqnarray}
The intersection above is over the convex hulls of all subsets of $Y$ of size $|Y|-f$.

 \vspace*{6pt}
\hrule
 \vspace*{2pt}
{\bf
Exact BVC algorithm} for $n\geq \max(3f+1,\, (d+1)f+1)$\,:
 \vspace*{4pt}
\hrule
\hrule
\begin{enumerate}
\item Each process uses a scalar {\em Byzantine broadcast} algorithm 
(such as \cite{lamport82,dolev1990})
to broadcast each element of its input vector to all the other processes
(each element is a scalar).
The {\em Byzantine broadcast} algorithm allows a designated sender to broadcast a scalar
value to the other processes, while satisfying the following properties when $n\geq 3f+1$:
(i) all the non-faulty processes decide on an identical scalar value, and (ii) if
the sender is non-faulty, then the value decided by the non-faulty processes is the sender's proposed
(scalar) value.
Thus, non-faulty processes can agree on the $d$ elements of the input vector at each
of the $n$ processes.

 At the end of the this step,
 each non-faulty process would have received an {\em identical} multiset $S$
 containing $n$ vectors, such that the vector corresponding to each non-faulty
 process is identical to the input vector at that process.
\item 
Each process chooses as its {\em decision} vector
a point in $\I$; all non-faulty processes choose the point identically using
a deterministic function.
We will soon show that $\I$ is non-empty.

\end{enumerate}
\hrule

We now prove that the above algorithm is correct.
Later, we show how
the {\em decision vector} can be found in Step 2 using linear programming.
The proof of correctness of the above algorithm
uses the following celebrated theorem by Tverberg \cite{perles07}:
\begin{theorem}
[Tverberg's Theorem \cite{perles07}]
\label{t_tverberg}
For any integer $f \geq 1$, and for every multiset $Y$
containing at least $(d+1)f+1$ points in ${\bf R}^d$, there exists a partition $Y_1,\cdots, Y_{f+1}$
of $Y$ into $f+1$ non-empty multisets such that $\cap_{l=1}^{f+1}\, \scripth(Y_l)\neq \emptyset$.
\end{theorem} 
The points in multiset $Y$ above are not necessarily distinct \cite{perles07};
thus, the same point may occur multiple times in $Y$.
(Appendix \ref{a_multisets} elaborates on the notion of multisets, and multiset partition.)
The partition in Theorem \ref{t_tverberg} is called a {\em Tverberg partition}, and the points in
$\cap_{l=1}^{f+1}\, \scripth(Y_l)$ in Theorem \ref{t_tverberg} are called {\em Tverberg points}.
Appendix \ref{a_tverberg} provides an illustration of a
Tverberg partition for points in 2-dimensional space.

The lemma below is used to prove the correctness of the above algorithm, as well
as the algorithm presented later in Section \ref{s_async}.

\begin{lemma}
For any multiset $Y$
containing at least $(d+1)f+1$ points in ${\bf R}^d$,
$\Gamma(Y)\neq\emptyset$.
\label{lem:TverSuper}
\end{lemma}
\begin{proof}
Consider a Tverberg partition
of $Y$ into $f+1$ non-empty subsets $Y_1,\cdots,Y_{f+1}$,
such that the set of Tverberg points
$\cap_{l=1}^{f+1}\, \scripth(Y_l)\neq\emptyset$.
Since $|Y|\geq (d+1)f+1$, by Theorem \ref{t_tverberg}, such a partition exists.
By (\ref{e_I}) we have
\begin{eqnarray}
\Gamma(Y) = \cap_{T\subseteq Y,\, |T|=|Y|-f}~\scripth(T).\label{e_I_2}
\end{eqnarray}
Consider any $T$ in (\ref{e_I_2}).
Since $|T|=|Y|-f$ and there are $f+1$ subsets in the Tverberg partition of $Y$,
$T$ excludes elements from at most $f$ of these subsets.
Thus, $T$ contains at least one subset from the partition.
Therefore, for {\bf each} $T$,
$\cap_{l=1}^{f+1}\, \scripth(Y_l)~\subseteq~\scripth(T)$.
Hence, from (\ref{e_I_2}), it follows that
$ \cap_{l=1}^{f+1}\, \scripth(Y_l)~ \subseteq~\Gamma(Y)$. 
Also, because
$ \cap_{l=1}^{f+1}\, \scripth(Y_l)\neq\emptyset$,
it now follows that $\Gamma(Y)\neq\emptyset$. 
\end{proof}

\comment{++++++++++++++++++++++++
\begin{lemma}
For any multiset $Y$
containing at least $(d+1)f+1$ points in ${\bf R}^d$,
consider a Tverberg partition $Y_1,\cdots,Y_{f+1}$.
Then, $\Gamma(Y)\neq\emptyset$, and all the Tverberg points
for the Tverberg partition are
contained in $\Gamma(Y)$,
that is, $\cap_{l=1}^{f+1}\,\scripth(Y_l)~\subseteq~\Gamma(Y)$.
\label{lem:TverSuper}
\end{lemma}
\begin{proof}
Consider Tverberg partition of $Y$ into $f+1$ non-empty multisets $Y_1,\cdots,Y_{f+1}$.
As per Theorem \ref{t_tverberg}, the set of Tverberg points
$\cap_{l=1}^{f+1}\, \scripth(Y_l)\neq\emptyset$.
By (\ref{e_I}) we have
\begin{eqnarray}
\Gamma(Y) = \cap_{T\subset Y,\, |T|=|Y|-f}~\scripth(T).\label{e_I_2}
\end{eqnarray}
Consider any $T$ in (\ref{e_I_2}).
Since $|T|=|Y|-f$ and there are $f+1$ subsets in the partition of $Y$,
$T$ excludes elements from at most $f$ of these subsets.
Thus, $T$ contains at least one subset from the partition.
Therefore, for {\bf each} $T$,
$\cap_{l=1}^{f+1}\, \scripth(Y_l)~\subseteq~\scripth(T)$.
Hence, from (\ref{e_I_2}), it follows that
$ \cap_{l=1}^{f+1}\, \scripth(Y_l)~ \subseteq~\Gamma(Y)$. 
Also, because
$ \cap_{l=1}^{f+1}\, \scripth(Y_l)\neq\emptyset$,
it now follows that $\Gamma(Y)\neq\emptyset$. 
\end{proof}
++++++++++++++++++++++++++++++++++++++++++}

We can now prove the correctness of our Exact BVC algorithm.
\begin{theorem}
\label{t_exact_suff}
$n\geq \max(3f+1,(d+1)f+1)$ is sufficient for achieving Exact BVC in a synchronous system.
\end{theorem}
\begin{proof}
We prove that the above {\em Exact BVC} algorithm is correct when $n\geq \max(3f+1,(d+1)f+1)$. 
The {\em termination} condition holds because the {\em Byzantine broadcast} algorithm used in
Step 1 terminates in finite time.
%
%
Since $|S|=n\geq (d+1)f+1$, by Lemma \ref{lem:TverSuper},
$\I\neq\emptyset$.
\comment{++++++++++++++++++++
Since $|S|=n\geq (d+1)f+1$, by Theorem \ref{t_tverberg}, 
there exists a Tverberg partition $S_1,S_2,\cdots,S_{f+1}$ of
$S$ with a non-empty set of Tverberg points, $J = \cap_{l=1}^{f+1}\, \scripth(S_l)$.
From Lemma \ref{lem:TverSuper}, $J \subseteq \I$.
Since $J\neq\emptyset$, $\I\neq\emptyset$.
++++++++++++++++++}
By (\ref{e_I}) we have
\begin{eqnarray}
\Gamma(S) = \cap_{T\subseteq S,\, |T|=|S|-f}~\scripth(T).\label{e_S_2}
\end{eqnarray}
At least one of the multisets $T$ in (\ref{e_S_2}), say $T^*$, must contain the
inputs of {\em only} non-faulty processes,
because $|T|=|S|-f=n-f$, and there are at most $f$ faulty processes.
By definition of $\I$, $\I\subseteq \scripth(T^*)$. Then, from the definition of
$T^*$, and the fact that the decision vector is chosen from $\I$, the {\em validity}
condition follows.
 
{\em Agreement} condition holds because all the non-faulty processes have identical
$S$, and pick as their decision vector a point in $\Gamma(S)$ using a deterministic function
in Step 2.
\end{proof}

We now show how Step 2 of the Exact BVC algorithm can be implemented 
using linear programming. The input to the linear program is 
$S=\{\bfs_i~:~1\leq i\leq n\}$, a multiset of $d$-dimensional vectors.
Our goal is to find a vector $\bfz \in \I$; or equivalently, find a vector $\bfz$
that can be expressed as a convex combination of vectors in $T$ for all choices $T \subseteq S$
such that $|T| = n-f$.
The linear program uses the following $d+\nchoosek{n}{n-f}(n-f)$ variables.
\begin{itemize}
\item $\bfz_1,..\bfz_d$\,: variables for $d$ elements of vector $\bfz$.
\item $\alpha_{T,i}$\,: coefficients such that $\bfz$ can be written as convex combination of vectors in $T$.
We include here only those $n-f$ indices $i$ for which $\bfs_i \in T$.
\end{itemize}

For every $T$, the linear constraints are as follows.
\begin{itemize}
\item $\bfz = \sum_{\bfs_i \in T}\, \alpha_{T,i}\,\bfs_i$ ~~($\bfz$ is a linear combination of $\bfs_i \in T$)
\item $\sum_{\bfs_i \in T}\, \alpha_{T,i} = 1$ ~~(The sum of all coefficients for a particular $T$ is $1$)  
\item $\alpha_{T,i}\,  \geq 0$ for all $\bfs_i \in T$.
\end{itemize}
For every $T$, we get $d+1+n-f$ linear constraints, yielding a total of
${n \choose n-f}(d+1+n-f)$ constraints
in $d+{n \choose n-f}(n-f)$ variables.
Hence, for any {\em fixed} $f$, a point in $\I$
can be found in polynomial time by solving a linear program with
the number of variables and constraints that are polynomial in $n$ and
$d$ (but not in $f$). However, when $f$ grows with $n$, the
computational complexity is high.

We note here that the above Exact BVC algorithm remains correct if the
non-faulty processes identically
choose {\em any point} in $\I$ as the decision vector.
In particular, 
as seen in the proof of Lemma \ref{lem:TverSuper}, all the Tverberg points
are contained in $\I$, therefore, one of the Tverberg points for multiset $S$ may be chosen
as the decision vector.
It turns out that,
for arbitrary $d$, currently there is no known algorithm with polynomial complexity to compute a
Tverberg point for a given multiset \cite{agarwal2004algorithms,miller2010approximate,mulzer2012approximating}.
However, in some restricted cases, efficient algorithms are
known (e.g., \cite{jadhav94}).

\section{Asynchronous Systems}
\label{s_async}

We develop a tight necessary and sufficient condition for {\em approximate} asynchronous BVC.

\subsection{Necessary Condition for Approximate Asynchronous BVC}
\label{ss_async_nec}

\begin{theorem}
\label{t_approx_nec}
$n\geq (d+2)f+1$ is necessary for approximate BVC in an asynchronous system.
\end{theorem}
\begin{proof}
We first consider the case of $f=1$. 
Suppose that a correct algorithm exists for $n=d+2$.
Denote by $\bfx_k$ the input vector at each process $p_k$.
Now consider a process $p_i$, where $1\leq i\leq d+1$. 
Since a correct algorithm must tolerate one failure, process $p_i$ must terminate in all executions in which process $p_{d+2}$ does not take any steps.
Suppose that all the processes are non-faulty, but process $p_{d+2}$ does not take
any steps until all the other processes terminate.
At the time when process $p_i$ terminates ($1\leq i\leq d+1$),
 it cannot distinguish between the following $d+1$ scenarios:
\begin{itemize}
\item Process $p_{d+2}$ has crashed: In this case, to satisfy the {\em validity} condition, the decision of process
$p_i$ must be in the convex hull of the inputs of processes
$p_1,p_2,\cdots,p_{d+1}$. That is, the decision vector must be in the convex hull of $X_i^{d+2}$ defined below.
\begin{eqnarray}
 X_i^{d+2} & = & \{\bfx_k~:~1\leq k\leq d+1\} \label{e_d+2}
\end{eqnarray}
$\bfx_{d+2}$ is not included above, because until process $p_i$ terminates,
$p_{d+2}$ does not take any steps (so $p_i$ cannot learn any information
about $\bfx_{d+2}$).

\item Process $p_j$ ($j\neq i$, $1\leq j\leq d+1$) is faulty, and process $p_{d+2}$ is slow, and hence $p_{d+2}$ has not taken any steps yet:
Recall that we are considering $p_i$ at the time when it terminates. Since
process $p_{d+2}$ has not taken any steps yet, process $p_i$ cannot have
any information about the input at $p_{d+2}$. Also, 
in this scenario $p_j$ may be faulty, therefore, process $p_i$
cannot trust the correctness of the input at $p_j$.
Thus, to satisfy the validity condition,
the decision of process $p_i$ must be in the convex hull of $X_i^j$
defined below.
\begin{eqnarray}
 X_i^j & = & \{\bfx_k~:~k\neq j \mbox{~and~} 1\leq k\leq d+1\} \label{e_j}
\end{eqnarray}
\end{itemize}
The decision vector of process $p_i$ must be valid independent of which of the
above $d+1$ scenarios actually occurred.
Therefore, observing that $\scripth(X_i^{d+2})\supseteq \scripth(X_i^j)$,
where $j\neq i$,
we conclude that the decision vector must be in
\begin{eqnarray}
\cap_{j\neq i, 1\leq j\leq d+1}~ \scripth(X_i^j)
\label{e_valid}
\end{eqnarray}
%
Recall that $\epsilon>0$ is the parameter of the $\epsilon$-agreement condition
in Section \ref{s_intro}.
For $1\leq i\leq d$, suppose that the $i$-th element of input vector $\bfx_i$ is $4\epsilon$, and the remaining $d-1$ elements are 0. Also
suppose that $\bfx_{d+1}$ and $\bfx_{d+2}$ are both equal to the all-0 vector. 

Let us consider process $p_{d+1}$.
In this case, $\scripth(X_{d+1}^j)$ for $j\leq d$ only contains
vectors whose $j$-th element is 0. Thus, the intersection of
all the convex hulls in (\ref{e_valid}) only contains the all-0 vector,
which, in fact, equals $\bfx_{d+1}$.
%
Thus, the decision vector of process $p_{d+1}$ must be equal to $\bfx_{d+1}$.
We can similarly show that for each $p_i$, $1\leq i\leq d+1$, 
the intersection in (\ref{e_valid}) only contains vector $\bfx_i$,
and therefore,
the decision vector
of process $p_i$ must be equal to its input $\bfx_i$.
The input vectors at each pair of processes in $p_1,\cdots,p_{d+1}$ differ by
$4\epsilon$ in at least one element. This implies that the $\epsilon$-agreement condition
is not satisfied.
Therefore, $n=d+2$ is not sufficient for $f=1$.
It should be easy to see that $n\leq d+2$ is also not sufficient.

For the case when $f>1$, by using a {\em simulation} 
similar to the proof of Theorem \ref{t_exact_nec}, 
we can now show that $n\leq (d+2)f$ is not sufficient.
Thus, $n\geq (d+2)f+1$ is necessary for $f\geq 1$.
(For $f=0$, the necessary condition holds trivially.)
\end{proof}

\subsection{Sufficient Condition for Approximate Asynchronous BVC}
\label{ss_async_suff}

We will prove that $n\geq (d+2)f+1$ is sufficient by proving the
correctness of an algorithm presented in this section.
The proposed algorithm
executes in asynchronous rounds. Each process $p_i$
maintains a local state $\bfv_i$, which is a $d$-dimensional vector.
We will refer to the value of $\bfv_i$
at the {\em end} of the $t$-th round performed by process $p_i$ as $\bfv_i[t]$.
Thus, $\bfv_i[t-1]$ is the value of $\bfv_i$ at the {\em start}
of the $t$-th round of process $p_i$.
The initial value of $\bfv_i$, namely $\bfv_i[0]$, is equal to $p_i$'s {\em input} vector,
denoted as $\bfx_i$.
The messages sent by each process anytime during its $t$-th round are tagged
by the round number $t$.
This allows a process $p_i$ in its round $t$ to determine,
despite the asynchrony, whether
a message received from another process $p_j$ was sent by $p_j$ in 
$p_j$'s round $t$.

The proposed algorithm is obtained by suitably modifying a {\em scalar} consensus
algorithm presented by Abraham, Amit and Dolev \cite{Abraham_optimalresilience04}
to achieve asynchronous approximate Byzantine scalar consensus among $3f+1$ processes.
We will refer to the algorithm in \cite{Abraham_optimalresilience04} as the AAD
algorithm. We first present a brief overview of the AAD algorithm, and describe
its properties. We adopt our notation above when describing the AAD algorithm
(the notation differs from \cite{Abraham_optimalresilience04}).  
One key difference is that, in our proposed algorithm $\bfv_i[t]$ is a vector, whereas in
AAD description below, it is considered a scalar. 
The AAD algorithm may be viewed as consisting of three components:
\begin{enumerate}
\item {\em AAD component \#1:} In each round $t$, the AAD algorithm requires each process
to communicate its state $\bfv_i[t-1]$ to other processes using a mechanism
that achieves the 
properties described next.
AAD ensures
that each non-faulty process $p_i$ in its round $t$ obtains a set $B_i[t]$
containing at least $n-f$ tuples of the form $(p_j,\bfw_j,t)$,
such that
the following properties hold:
\begin{itemize}
\item (Property 1)~~ For any two non-faulty processes $p_i$ and $p_j$:
\begin{eqnarray}
\label{e_overlap}
|B_i[t]\cap B_j[t]|\,\geq\, n-f
\end{eqnarray}
That is, $p_i$ and $p_j$ learn at least $n-f$ identical tuples.
\item (Property 2)~~ If $(p_l,\bfw_l,t)$ and $(p_k,\bfw_k,t)$ are both in $B_i[t]$, then $p_l\neq p_k$. That is,
	$B_i[t]$ contains at most one tuple for each process.
\item (Property 3)~~ If $p_k$ is non-faulty, and $(p_k,\bfw_k,t)\in B_i[t]$, then $\bfw_k=\bfv_k[t-1]$.
	That is, for any non-faulty process $p_k$, $B_i[t]$ may only contain 
	the tuple $(p_k,\bfv_k[t-1],t)$. (However, it is possible that, corresponding
	to some non-faulty process, $B_i[t]$ does not contain a tuple at all.)
\end{itemize}

\item {\em AAD component \#2:} Process $p_i$, having obtained set $B_i[t]$ above, computes its new state $\bfv_i[t]$
	as a function of the tuples in $B_i[t]$.
	The primary difference between our proposed algorithm and AAD is in this step.
	The computation of $\bfv_i[t]$ in AAD is designed to be correct for scalar inputs (and scalar decision),
	whereas our approach applies to $d$-dimensional vectors.

\item {\em AAD component \#3:} AAD also includes a sub-algorithm that allows the non-faulty processes to determine
	when to terminate their computation. Initially, the processes 
	cooperate to
	estimate a quantity $\delta$ as a function of the input
	values at various processes.
 Different non-faulty processes may estimate different values for $\delta$,
 	since the estimate  is affected by the behavior of faulty processes and message delays.
	Each process then uses $1+\lceil\log_2\frac{\delta}{\epsilon}\rceil$ as the
	threshold
on the minimum number of rounds necessary for the non-faulty processes
	to converge within $\epsilon$ of each other.
	The base of the logarithm above is 2,
	because the range of the values at
	the non-faulty processes is shown to shrink by a factor of $\frac{1}{2}$
	after each asynchronous round of AAD \cite{Abraham_optimalresilience04}.
        Subsequently, when the processes reach respective thresholds on the rounds,
        they exchange additional messages. After an adequate number of processes announce
	that they have reached their threshold, 
	all the non-faulty processes may terminate.
\end{enumerate}

It turns out that the Properties 1, 2 and 3 hold even if {\em Component \#1}
of AAD is used with $\bfv_i[t]$ as a {\em vector}. We exploit these properties in our algorithm
below.
\noindent
The proposed algorithm below uses a function $\Phi$, which takes a
set, say set $B$, containing tuples of the form $(p_k,\bfw_k,t)$, and
returns a multiset containing the points (i.e., $\bfw_k$). Formally,
\begin{eqnarray}
\Phi(B) & = & \{ \bfw_k ~ : ~ (p_k,\bfw_k,t)\in B\} \label{e_Phi}
\end{eqnarray}

A mechanism similar to that in AAD may potentially be used to achieve termination for the approximate BVC algorithm below as well. The main difference
from AAD would be in the manner in which the threshold on the number of
rounds necessary is computed. However, for brevity, we simplify
our algorithm by 
assuming that there exists an upper bound $U$ and a lower bound
$\nu$ on the values of the $d$ elements in the inputs vectors
at non-faulty processes, and that these bounds are known 
{\em a priori}. Thus, all the elements in each input vector will be $\leq U$
and $\geq \nu$. This assumption holds in many practical systems, because the
input vector elements represent quantities that are constrained. For instance,
if the input vectors are probability vectors, then $U=1$ and $\nu=0$.
If the input vectors represent locations in 3-dimensional space occupied
by mobile robots, then $U$ and $\nu$ are determined by the boundary of
the region in which the robots are allowed to operate.
The advantage of the AAD-like solution over our simple approach is
that, depending on the actual inputs, the algorithm may potentially terminate
sooner, and the AAD mechanism prevents faulty processes from causing the
non-faulty processes to run longer than necessary. However, the simple static
approach for termination presently suffices to prove the correctness of our
approximate BVC algorithm, as shown later.

\vspace*{6pt}
\hrule \vspace*{2pt}
{\bf
Asynchronous Approximate BVC algorithm
}
 for $n\geq (d+2)f+1$\,:
\vspace*{4pt}
\hrule
\begin{enumerate}
\item In the $t$-th round, each non-faulty process uses the mechanism in
{\em Component \#1} of the AAD algorithm to obtain
a set $B_i[t]$ containing at least $n-f$ tuples, such that $B_i[t]$ satisfies properties 1, 2, and 3 described
earlier for AAD.
While these properties
were proved in \cite{Abraham_optimalresilience04} for scalar states, the correctness of the properties
also holds when $\bfv_i$ is a vector.

\item In the $t$-th round, after obtaining set $B_i[t]$, process $p_i$ computes its new
	state $\bfv_i[t]$ as follows. Form a multiset $Z_i$ using the steps below:
	\begin{itemize}
	\item Initialize $Z_i$ as empty.
	\item For each $C\subseteq B_i[t]$ such that $|C|=n-f\geq (d+1)f+1$,
		add to $Z_i$ one deterministically chosen point from $\Gamma(\Phi(C))$.
			Since $|\Phi(C)|=|C|\geq(d+1)f+1$, by Lemma \ref{lem:TverSuper}, $\Gamma(\Phi(C))$ is non-empty. 
	\end{itemize}
Note that $|Z_i|~=~{|B_i[t]| \choose n-f}~\leq~ {n \choose n-f}$.
Calculate \begin{eqnarray}
\bfv_i[t]~=~\frac{ \sum_{\bfz\in Z_i}~ \bfz }{|Z_i|}\label{e_algo_Z}
\end{eqnarray}

\item Each non-faulty process terminates after $1+\lceil\log_{1/(1-\gamma)}\,
	\frac{U-\nu}{\epsilon}\rceil$ rounds, where $\gamma$ ($0<\gamma<1$) is a
	constant defined later in (\ref{e_gamma}).
	Recall that $\epsilon$ is the parameter of the $\epsilon$-agreement
	condition.
\end{enumerate}
\hrule

In Step 2 above, we consider
$\nchoosek{|B_i[t]|}{n-f}$
subsets $C$ of $B_i[t]$, each subset being of size $n-f$.
As elaborated in Appendix \ref{a_opt}, it is
possible to reduce the number of subsets explored to just $n-f$.
This optimization
will reduce the computational complexity of Step 2, but it is not necessary
for correctness of the algorithm.

\begin{theorem}
\label{t_approx_suff}
$n\geq (d + 2)f + 1$ is sufficient for approximate BVC in an
asynchronous system.
\end{theorem}
\begin{proof}
Without loss of generality, suppose that $m$ processes $p_1,p_2,\cdots p_m$
are non-faulty, where $m\geq n-f$, and the remaining $n-m$ processes 
are faulty. \vspace*{4pt}
In the proof,
we will often omit the round index $[t]$ in $B_i[t]$, since the index should
be clear from the context.
In this proof, we consider the steps taken by the non-faulty processes in their respective
$t$-th rounds, where $t>0$.
We now define a {\em valid} point. The definition is used later in the proof.
\begin{definition}
\label{d_valid}
A point $\bfr$ is said to be {\em valid} if there exists a representation
of $\bfr$ as a convex combination of $\bfv_k[t-1]$, $1\leq k\leq m$.
That is, there exist constants $\beta_k$, such that $0\leq \beta_k\leq 1$
and $\sum_{1\leq k\leq m}~\beta_k = 1$, and
\begin{eqnarray}
\bfr & = & \sum_{1\leq k\leq m}~ \beta_k\,\bfv_k[t-1]
\label{e_cc}
\end{eqnarray}
$\beta_k$ is said to be the {\bf weight} of $\bfv_k[t-1]$ in the above convex combination.
\end{definition}
In general, there may exist multiple such convex combination representations of a
{\em valid}\, point $\bfr$. 
Observe that at least one of the weights in any such convex combination must be $\geq\frac{1}{m}
\geq \frac{1}{n}$.

For the convenience of the readers,
we break up the rest of this proof into three parts.
\paragraph{Part I:}

At a non-faulty process $p_i$,
consider any $C\subseteq B_i$ such that $|C|=n-f$ (as in Step 2 of the algorithm).
Since $|\Phi(C)|=|C|=n-f\geq (d+1)f+1$,
by Lemma \ref{lem:TverSuper}, 
$\Gamma(\Phi(C))\neq\emptyset$.
So $Z_i$ will contain a point from $\Gamma(\Phi(C))$ for each $C$.


Now, $C\subseteq B_i$,
$|\Phi(C)|=n-f$, and there are at most $f$ faulty processes.
Then
Property 3 of $B_i$ implies that
at least one $(n-2f)$-size subset
of $\Phi(C)$ must also be a subset of $\{\bfv_1[t-1],\bfv_2[t-1],\cdots,\bfv_m[t-1]\}$,
i.e., contain only the state of non-faulty processes. Therefore, 
all the points in $\Gamma(\Phi(C))$ must be {\em valid} (due to (\ref{e_I})
and Definition \ref{d_valid}).
This observation is true for each set $C$ enumerated in Step 2.
Therefore, all the points in $Z_i$ computed in Step 2 must be valid.
(Recall that we assume 
processes $p_1,\cdots,p_m$ are non-faulty.)

\vspace*{6pt}
\hrule
\paragraph{Part II:}
Consider any two non-faulty processes $p_i$ and $p_j$.

\begin{itemize}
\item{\em Observation 1:}
As argued in Part I, all the points in $Z_i$ are valid.
Therefore, all the points in $Z_i$ can be expressed as convex combinations
of the state of non-faulty processes, i.e., $\{\bfv_1[t-1],\cdots,\bfv_m[t-1]\}$.
Similar observation holds for all the points in $Z_j$ too.

\item{\em Observation 2:}
By Property 1 of $B_i$ and $B_j$,\footnote{As noted earlier, we omit the
round index $[t]$ when discussing the sets $B_i[t]$ and $B_j[t]$ here.}
 \[ |B_i\cap B_j|\,\geq\, n-f.\]
Therefore, there exists a set $C_{ij}\subseteq B_i\cap B_j$
such that $|C_{ij}|=n-f$.
Therefore, $Z_i$ and $Z_j$ both contain one identical point from $\Gamma(\Phi(C_{ij}))$.
Suppose that this point is named $\bfz_{ij}$.
As shown in Part I above, $\bfz_{ij}$ must be {\em valid}. 
Therefore, there exists a convex combination representation of $\bfz_{ij}$ in terms of the
states $\{\bfv_1[t-1],\bfv_2[t-1],\cdots,\bfv_m[t-1]\}$ of non-faulty processes.
Choose any one such convex combination.
There must exist a non-faulty process, say $p_{g(i,j)}$, 
such that the weight associated with $\bfv_{g(i,j)}[t-1]$ in the
convex combination
for $\bfz_{ij}$ is $\geq\frac{1}{m}\geq\frac{1}{n}$.
We can now make the next observation.\footnote{Note that, to simplify
the notation somewhat, the notation $g(i,j)$ does not make the round index $t$ explicit.
However, it should be noted that $g(i,j)$ for processes $p_i$ and $p_j$
can be different in different rounds. 
}

\item {\em Observation 3:}
Recall from (\ref{e_algo_Z}) that $\bfv_i[t]$ is computed as the average of the points in $Z_i$, 
and $|Z_i|={|B_i|\choose n-f}\leq {n\choose n-f}$.
By {\em Observations 1},
all the points in $Z_i$ are valid, and by {\em Observation 2},
$\bfz_{ij}\in Z_i$.
These observations together imply that
$\bfv_i[t]$ is also valid, and
{\em there exists} a representation of $\bfv_i[t]$
as a convex combination 
of $\{\bfv_1[t-1],\cdots,\bfv_m[t-1]\}$,
wherein the weight of $\bfv_{g(i,j)}[t-1]$
is $\geq \frac{1}{n\, \nchoosek{|B_i|}{n-f}}\geq\frac{1}{n\,\nchoosek{n}{n-f}}$.
Similarly, we can show that 
{\em there exists} a representation
of $\bfv_j[t]$ as a convex combination of
$\{\bfv_1[t-1],\cdots,\bfv_m[t-1]\}$,
wherein 
the weight of $\bfv_{g(i,j)}[t-1]$
is $\geq \frac{1}{n\,\nchoosek{n}{n-f}}$.
Define
\begin{eqnarray} \gamma = \frac{1}{n\, \nchoosek{n}{n-f}} \label{e_gamma} \end{eqnarray}
Consensus is trivial for $n=1$, so we consider finite $n>1$.
Therefore, $0<\gamma<1$.
\end{itemize}

\vspace*{2pt}
\hrule
\paragraph{Part III:}

Observation 3 above implies that for any $\tau>0$,
$\bfv_i[\tau]$ is a convex combination of
$\{\bfv_1[\tau-1],\cdots,\bfv_m[\tau-1]\}$. Applying this observation for
$\tau=1,2,\cdots,t$,
we can conclude that $\bfv_i[t]$ is a convex combination of $\{\bfv_1[0],
\cdots,\bfv_m[0]\}$, implying that the proposed algorithm satisfies the
{\bf validity} condition for approximate consensus. (Recall that
$\bfv_k[0]$ equals process $p_k$'s input vector.)

Let $\bfv_{il}[t]$ denote the $l$-th element of the vector state $\bfv_i[t]$ of
process $p_i$.
Define $\Omega_l[t] = \max_{1\leq k\leq m}~ \bfv_{kl}[t]$,
the maximum value of $l$-th element of the vector state of non-faulty processes.
Define $\mu_l[t] = \min_{1\leq k\leq m}~ \bfv_{kl}[t]$, 
the minimum value of $l$-th element of the vector state of non-faulty processes.
Appendix \ref{a_agreement} proves, using 
{\em Observations 1} and {\em 3} above, that
\begin{eqnarray}
\Omega_l[t] - \mu_l[t] & \leq & (1-\gamma)\, \left(\Omega_l[t-1]-\mu_l[t-1]\right),~~~~
\mbox{for}~ 1\leq l\leq d
\label{e_convergence}
\end{eqnarray}
By repeated application of (\ref{e_convergence})
we get
\begin{eqnarray}
\Omega_l[t] - \mu_l[t] & \leq & (1-\gamma)^t\, \left(\Omega_l[0]-\mu_l[0]\right)
\end{eqnarray}
Therefore, for a given $\epsilon>0$, if 
\begin{eqnarray} t&>&\log_{1/(1-\gamma)}~~ \frac{
\Omega_l[0]-\mu_l[0]
}{\epsilon},
\label{e_t_threshold}
\end{eqnarray}
then
\begin{eqnarray} \Omega_l[t]-\mu_l[t] ~ <~\epsilon. \label{e_rho_epsilon} \end{eqnarray}
Since (\ref{e_t_threshold}) and (\ref{e_rho_epsilon}) hold for $1\leq l\leq d$,
and $U\geq \Omega_l[0]$ and $\nu\leq \mu_l[0]$ for $1\leq l\leq d$,
if each non-faulty process
terminates after $1+\lceil\log_{1/(1-\gamma)}\, \frac{U-\nu}{\epsilon}\rceil$
rounds, $\epsilon$-agreement is ensured. As shown previously, validity condition
is satisfied as well.
Thus, the proposed algorithm is correct, and $n\geq (d+2)f+1$ is
sufficient for approximate consensus in asynchronous systems.
\end{proof}

~

\section{Simpler Approximate BVC Algorithms with Restricted Round Structure}
\label{s_simple}

The proposed approximate BVC algorithm relies on {\em Component \#1}
of AAD for exchange of state information among the processes.
The communication pattern of AAD requires three message delays in each
round (i.e., a causal chain of three messages per round), to ensure
strong properties for sets $B_i[t]$, as summarized in Section \ref{ss_async_suff}.
In this section, we consider  simpler (restricted) round structure
that reduces the communication delay, and the number of messages, per round.
The price of the reduction in message cost/delay is an increase in the
number of processes necessary to achieve approximate BVC, as seen below.

We consider a restricted round structure 
for achieving {\em approximate} consensus in
{\em synchronous} and {\em asynchronous} settings both.
In both settings, 
each process $p_i$ maintains state $\bfv_i[t]$, as in the case
of the algorithm in Section \ref{ss_async_suff}.
$\bfv_i[0]$ is initialized to the input vector at process $p_i$.

\paragraph{\normalfont\em Synchronous approximate BVC:}
The restricted algorithm structure for a synchronous system
is as follows. The algorithm executes in synchronous rounds, and each process
$p_i$ performs the following steps in the $t$-th round, $t>0$.
\begin{enumerate}
\item Transmit current vector state, $\bfv_i[t-1]$, to all the processes.
Receive vector state from all the processes. If a message is not
received from some process, then its vector state is assumed to have some
default value (e.g., the all-0 vector).

\item Compute new state $\bfv_i[t]$ as a function
of $\bfv_i[t-1]$ and the vectors received from the other processes in the above step.
\end{enumerate}

\paragraph{\normalfont\em Asynchronous approximate BVC:}
The restricted structure of the asynchronous rounds in the asynchronous setting is
similar to that in \cite{AA_Dolev_1986}. The messages in this case are
tagged by the round index, as in Section \ref{ss_async_suff}.
Each process $p_i$ performs the following steps in its $t$-th round, $t>0$: 
\begin{enumerate}
\item Transmit current state $\bfv_i[t-1]$ to all the processes.
These messages are tagged by round index $t$.

Wait until a message tagged by round index $t$ is received
from $(n-f-1)$ {\em other} processes.

\item Compute new state $\bfv_i[t]$ as a function of $\bfv_i[t-1]$,
and the $(n-f-1)$ other vectors 
collected in the previous step (for a total of $n-f$ vectors).
\end{enumerate}

\noindent
For algorithms with the above round structures, the following results
can be proved; the proofs are similar to those in Section \ref{s_async}.

\begin{theorem}
\label{t_simple}
For the restricted synchronous and asynchronous round structures
presented above in Section \ref{s_simple},
following conditions are necessary and sufficient:
\begin{itemize}
\item Synchronous case: $n\geq (d+2)f+1$ 
\item Asynchronous case: $n\geq (d+4)f+1$ 
\end{itemize}
\end{theorem}
To avoid repeating the ideas used in Section \ref{s_async}, we do not present
complete formal proofs here.
We can prove sufficiency constructively. The restricted round
structures above already specify the Step 1 of each round.
We can use Step 2 analogous to that of the algorithm in Section \ref{ss_async_suff},
with $B_i[t]$ being redefined as the set of vectors received by process $p_i$ in Step 1 of
the restricted structure.
\begin{itemize}

\item In the synchronous setting, $n\geq (d+2)f+1$ is necessary.
With $n\geq (d+2)f+1$, observe that any two non-faulty processes
$p_i$ and $p_j$ will receive identical vectors from $n-f\geq (d+1)f+1$
non-faulty processes. 
Thus, $B_i[t]\cap B_j[t]$ contains at least $(d+1)f+1$ identical
vectors.

\item In the asynchronous setting, $n\geq (d+4)f+1$ is necessary.
With $n\geq (d+4)f+1$,
each non-faulty processe will have, in Step 2, vectors
from at least $n-f$ processes (including itself). Thus, any two fault-free processes
will have, in Step 2, vectors from at least $n-2f$ identical
processes, of which at most $f$ may be faulty. Thus,
$B_i[t]\cap B_j[t]$ contains at least $n-3f$ identical vectors (corresponding
to the state of $n-3f$ non-faulty processes). Note that $n-3f\geq (d+1)f+1$.
\end{itemize}
The proof of correctness of the algorithm in Section \ref{ss_async_suff}
relies crucially on the property that \[|B_i[t]\cap B_j[t]|~\geq~(d+1)f+1.\]
As discussed above, when the number of nodes satifies the constraints
in Theorem \ref{t_simple}, this property holds for the restricted round
structures too. The rest of the proof of correctness of the
restricted algorithms is then similar to the proof of Theorem \ref{t_approx_nec}. 
Thus, the above synchronous and asynchronous algorithms can achieve
{\rm approximate} BVC.

\section{Summary}
\label{s_summary}

This paper addresses Byzantine vector consensus (BVC) wherein the input
at each process, and its decision, is a $d$-dimensional vector. We derive tight necessary
and sufficient bounds on 
the number of processes required for {\em Exact BVC} in synchronous systems, and
{\em Approximate BVC} in asynchronous systems.

In Section \ref{s_simple}, we derive bounds on 
the number of processes required for algorithms
with restricted round structures to achieve {\em approximate}
consensus in synchronous as well as
asynchronous systems.


\section*{Acknowledgments}

 Nitin Vaidya acknowledges Eli Gafni for suggesting
the problem of vector consensus, Lewis Tseng for feedback, and Jennifer Welch for answering
queries on distributed computing. Vijay Garg acknowledges John Bridgman and Constantine Caramanis for
discussions on the problem.





\newpage


\appendix

\section*{Appendix}

\section{Notations}
\label{a_notations}

This appendix summarizes some of the notations and terminology introduced in the paper.

\begin{itemize}

\item $n$ = number of processes.

\item $\scriptp=\{p_1,p_2,\cdots,p_n\}$ is the set of processes in the system.

\item $f$ = maximum number of Byzantine faulty processes.

\item $d$ = dimension of the input vector as well as decision vector at each process.

\item $\bfx_i$ = $d$-dimensional input vector at process $p_i$.
The vector is equivalently viewed as a point in the Euclidean space ${\bf R}^d$.

\item $\scripth(Y)$ denotes the convex hull of the points in multiset $Y$.


\item $m$ : The proof of Theorem \ref{t_approx_suff} assumes, without loss of generality,
	that for some $m\geq n-f$, 
      processes $p_1,\cdots,p_m$ are non-faulty, and the remaining $n-m$ processes 
	are faulty.

\item $\Gamma(.)$ is defined in (\ref{e_I}).
\item $\Phi(.)$ is defined in (\ref{e_Phi}).

\item $\bfv_i[t]$ is the state of process $p_i$ at the end of its $t$-th
	round of the asynchronous BVC algorithm, $t>0$. Thus,
$\bfv_i[t-1]$ is the state of process $p_i$ at the start of its $t$-th
	round, $t>0$.
	$\bfv_i[0]$ for process $p_i$ equals its input $\bfx_i$.

\item $\bfv_{il}[t]$ is the $l$-th element of $\bfv_i[t]$, where $1\leq l\leq d$.

\item $B_i[t]$ defined in Section \ref{ss_async_suff}, is a
set of tuples of the form $(p_j,\bfw_j,t)$,
obtained by process $p_i$ in Step 1 of the
approximate consensus algorithm.

\item {\em Weight} in a convex combination is defined in Definition \ref{d_valid}

\item $\gamma = \frac{1}{n\nchoosek{n}{n-f}}$, as defined in (\ref{e_gamma}).
	Note that $0<\gamma<1$ for finite $n>1$.


\item $\Omega_l[t] = \max_{1\leq k\leq m}~ \bfv_{kl}[t]$

\item $\mu_l[t] = \min_{1\leq k\leq m}~ \bfv_{kl}[t]$

\item $\rho_l[t] = \Omega_l[t]-\mu_l[t]$

\item $|Y|$ denotes the size of a {\em multiset} $Y$.

\item $\parallel a\parallel$ is the absolute value of a real number $a$.

\end{itemize}

\section{Multisets and Multiset Partition}
\label{a_multisets}

Multiset is a generalization on the notion of a set.
While the members in a set must be distinct,
a multiset may contain the same member multiple times.

Notions of a {\em subset of a multiset} and a {\em partition of a multiset}
have natural definitions.
For completeness, we present the definitions here.

Suppose that $Y$ is a multiset. $Y$ contains $|Y|$ members.
Denote the members in $Y$ as $y_i$, $1\leq i\leq |Y|$.
Thus, $Y = \{ y_1,y_2,\cdots,y_{|Y|} \}$.
Define set $N_Y= \{1,2,\cdots,|Y|\}$. Thus, $N_Y$ contains
integers from 1 to $|Y|$.
Since $Y$ is a multiset, it is possible that $y_i=y_j$ for some $i\neq j$.

$Z$ is a subset of $Y$ provided that
there exists a set $N_Z\subseteq N_Y$ such that
\[
Z = \{ y_i~:~ i\in N_Z\}
\]

Subsets $Y_1,Y_2,\cdots, Y_b$ of multiset $Y$ form a partition of $Y$ provided
that there exists a partition $N_1,N_2,\cdots, N_b$ of set $N_Y$
such that
\[
Y_j = \{ y_i~:~i \in N_j\}, ~~~~ 1\leq j\leq b
\] 

\section{Clarification for the Proof of Theorem \ref{t_exact_nec}}
\label{a_sync}

In the proof of Theorem \ref{t_exact_nec}, when considering
the case of $f=1$, we claimed the following:
\begin{list}{}{}
\item
Since none of the non-faulty processes know which process, if any, is faulty,
as elaborated in Appendix \ref{a_sync}, the
decision vector must be in the convex hull of each multiset containing
the input vectors of $n-1$ of the processes (there are $n$ such multisets).
Thus, this intersection must be non-empty,
for all possible input vectors at the $n$ processes.
\end{list}
Now we provide an explanation for the above claim.

Suppose that the input at process $p_i$ is $\bfx_i$, $1\leq i\leq n$.
All the processes are non-faulty, but the processes do not know this fact.
The decision vector chosen by the processes must 
satisfy the {\em agreement} and {\em validity} conditions both.
\begin{itemize}
\item With $f=1$, any one process may potentially be faulty. In particular,
process $p_i$ ($1\leq i\leq n$) may possibly be faulty. Therefore,
the input $\bfx_i$ of process $p_i$ cannot be trusted by other processes.
Then to ensure {\em validity}, the decision vector chosen by any other process $p_j$ ($j\neq i$)
must be in the convex hull of the inputs at the processes
in $\scriptp-\{p_i\}$ \,(i.e., all processes except $p_i$).
Thus, the decision vector of process $p_j$ ($j\neq i$) must be in the
convex hull of the points in multiset $X^i$ below. \[X^i=\{\bfx_k~:~k\neq i,~1\leq k\leq n\}.\]

\item To ensure {\em agreement}, the decision vector chosen by all the processes must be identical. Therefore,
the decision vector must be in the intersection of the convex hulls of all the multisets $X^i$
($1\leq i\leq n$) defined
above. Thus,
 we conclude that the decision vector must be in the intersection below, where
 $\scripth(X^i)$ denotes the {\em convex hull}
of the points in multiset $X^i$, and $Q_i$ denotes $\scripth(X^i)$.
\begin{eqnarray} \cap_{i=1}^n~ \scripth(X^i) ~ =~\cap_{i=1}^n~Q_i\label{e_a_sync} \end{eqnarray}
\end{itemize}
If the intersection in (\ref{e_a_sync}) is empty, then there is no decision vector that satisfies {\em validity}
and {\em agreement} conditions both.
Therefore, the intersection must be non-empty.

As shown in the proof of Theorem \ref{t_exact_nec}, if $n$ is not large enough,
then the intersection in (\ref{e_a_sync}) may be empty.

\section{Tverberg Partition}
\label{a_tverberg}

Figure \ref{f_heptagon} illustrates a Tverberg partition of a
set of 7 vertices in 2-dimensions. The 7 vertices are at the corners
of a heptagon. Thus, $n=7$ here, and $d=2$.
Let $f=2$. Then, $n=(d+1)f+1$, and Tverberg's Theorem \ref{t_tverberg}
implies the presence of a Tverberg partition consisting of $f+1=3$ subsets.
Figure \ref{f_heptagon} shows the convex hulls of the three subsets
in the Tverberg partition: one convex hull is a triangle, and the other
two convex hulls are each a line segment.
In this example,
the three convex hulls intersect in exactly one point. Thus, there is
just one Tverberg point. In general, there can be multiple Tverberg points.

\begin{figure}[t]
\centerline{
\includegraphics[width=0.34\textwidth]{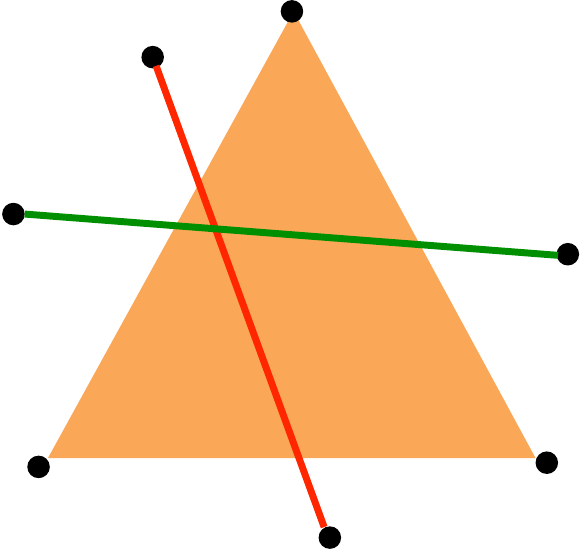}
}

\caption{\em Illustration of a Tverberg partition. \protect\newline
{\em Acknowledgment}: The above example is inspired by an illustration
authored by David Eppstein, which is available
in the public domain from {\em Wikipedia Commons}.\newline
}
\label{f_heptagon}
\end{figure}

\section{Proof of (\ref{e_convergence})}
\label{a_agreement}

$\bfv_{il}[t]$ denotes the $l$-th element of the vector state $\bfv_i[t]$ of
process $p_i$, $1\leq l\leq d$.
Processes $p_1,\cdots,p_m$ are non-faulty, and processes $p_{m+1},\cdots,p_n$
are faulty, where $m\geq n-f$.
Recall that, for $1\leq l\leq d$,
\begin{eqnarray}
\Omega_l[t] &  = & \max_{1\leq k\leq m}~ \bfv_{kl}[t], \mbox{~maximum value of $l$-th elements at non-faulty processes}\\
\mu_l[t] &  = & \min_{1\leq k\leq m}~ \bfv_{kl}[t],\mbox{~minimum value of $l$-th elements at non-faulty processes}\\
\mbox{Define}\hspace*{.5in} \\
\rho_l[t] & = & \Omega_l[t]~-~ \mu_l[t]
\end{eqnarray}
Equivalently,
\begin{eqnarray}
\rho_l[t] & = & \max_{1\leq i,j\leq m}~ \parallel\bfv_{il}[t]-\bfv_{jl}[t]\parallel
\label{rho_equi}
\end{eqnarray}
where $\parallel . \parallel$ operator yields the absolute value of the scalar parameter.

Consider any two non-faulty processes $p_i,p_j$ (thus,
$1\leq i,j \leq m$). Consider $1\leq l\leq d$. Then
\begin{eqnarray}
\mu_l[t-1] ~ \leq & \bfv_{il}[t-1] & \leq ~ \Omega_l[t-1] \\
\mu_l[t-1] ~ \leq & \bfv_{jl}[t-1] & \leq ~ \Omega_l[t-1] 
\end{eqnarray} 
{\em Observations 1} and {\em 3} in Part III of the proof of Theorem
\ref{t_approx_suff}, and the definition of $\gamma$,
imply the existence of constants $\alpha_k$'s
and $\beta_k$'s such that:
\begin{eqnarray}
\bfv_{i}[t] & = & \sum_{k=1}^m ~ \alpha_k\, \bfv_k[t-1] \mbox{~~~~where}
\label{e_i_alpha}\\
&& \alpha_k\geq 0 \mbox{~for}~1\leq k\leq m, \mbox{~~and~~} \sum_{k=1}^m\alpha_k=1 \\
&& \alpha_{g(i,j)}\geq \gamma\\~\nonumber\\
\bfv_{j}[t] & = & \sum_{k=1}^m ~ \beta_k\, \bfv_k[t-1] \mbox{~~~~where} \label{e_j_beta} \\
&& \beta_k\geq 0~\mbox{for~}1\leq k\leq m, \mbox{~~and~~} \sum_{k=1}^m\beta_k=1 \\
&& \beta_{g(i,j)}\geq \gamma \\ \nonumber \\ \nonumber
\end{eqnarray}

In the following, let us abbreviate $g(i,j)$ simply as $g$.
Thus, $\alpha_{g(i,j)}$ is same as $\alpha_g$,
and $\beta_{g(i,j)}$ is same as $\beta_g$.
From (\ref{e_i_alpha}) and (\ref{e_j_beta}),
focussing on just the operations on $l$-th elements,
 we obtain
\begin{eqnarray}
\bfv_{il}[t] & = & \sum_{k=1}^m ~ \alpha_k\, \bfv_{kl}[t-1] \nonumber\\
& \leq & \alpha_g\, \bfv_{gl}[t-1] ~+~(1-\alpha_g)\,\Omega_l[t-1] 
\mbox{~~~~because $\bfv_{kl}[t-1]\leq \Omega_l[t-1],~\forall k$} \nonumber\\
& \leq & \gamma \, \bfv_{gl}[t-1] ~+~ (\alpha_g-\gamma)\bfv_{gl}[t-1]~+~(1-\alpha_g)\,\Omega_l[t-1]  \nonumber \\
& \leq & \gamma \, \bfv_{gl}[t-1] ~+~ (\alpha_g-\gamma)\Omega_l[t-1]~+~(1-\alpha_g)\,\Omega_l[t-1]  \nonumber \\
&& \hspace*{1in}\mbox{~~~~because
$\bfv_{gl}[t-1]\leq \Omega_l[t-1]$ and $\alpha_g\geq \gamma$
} \nonumber \\
& \leq & \gamma\,\bfv_{gl}[t-1]~+~(1-\gamma)\,\Omega_l[t-1] 
\label{e_il} \\ \nonumber\\ \nonumber\\
\bfv_{jl}[t] & = & \sum_{k=1}^m ~ \beta_k\, \bfv_{kl}[t-1] \nonumber\\
& \geq & \beta_g\, \bfv_{gl}[t-1] ~+~(1-\beta_g)\,\mu_l[t-1] 
\mbox{~~~~because $\bfv_{kl}[t-1]\geq \mu_l[t-1],~\forall k$} \nonumber\\
& \geq & \gamma \, \bfv_{gl}[t-1] ~+~ (\beta_g-\gamma)\bfv_{gl}[t-1]~+~(1-\beta_g)\,\mu_l[t-1] \nonumber\\
& \geq & \gamma \, \bfv_{gl}[t-1] ~+~ (\beta_g-\gamma)\mu_l[t-1]~+~(1-\beta_g)\,\mu_l[t-1] \nonumber \\ 
&& \hspace*{1in}\mbox{~~~~because $\bfv_{gl}[t-1]\geq \mu_l[t-1]$, and
$\beta_g\geq \gamma$}\nonumber\\
& \geq & \gamma\,\bfv_{gl}[t-1]~+~(1-\gamma)\,\mu_l[t-1]
\label{e_jl}\\
 ~\nonumber\\
\Rightarrow~~
\bfv_{il}[t] ~-~\bfv_{jl}[t] & \leq & (1-\gamma)\,(\Omega_l[t-1]-\mu_l[t-1])
\mbox{~~~~ subtracting (\ref{e_jl}) from (\ref{e_il})} \label{il_jl}
\end{eqnarray}
By swapping the role of $p_i$ and $p_j$ above, we can also show that
\begin{eqnarray}
\bfv_{jl}[t] ~-~\bfv_{il}[t] & \leq & (1-\gamma)\,(\Omega_l[t-1]-\mu_l[t-1])
\label{jl_il}
\end{eqnarray}
Putting (\ref{il_jl}) and (\ref{jl_il}) together, we obtain
\begin{eqnarray}
\parallel\bfv_{il}[t] ~-~\bfv_{jl}[t]\parallel & \leq & (1-\gamma)\,(\Omega_l[t-1]-\mu_l[t-1])
\mbox{~~because $\Omega_l[t-1]\geq\mu_l[t-1]$}\nonumber\\
	&\leq & (1-\gamma)\,\rho_l[t-1] \mbox{~~~~ by the definition of $\rho_l[t-1]$} \label{il_jl_rho}\\ ~ \nonumber \\
\Rightarrow ~~~~
\max_{1\leq i,j\leq m}~\parallel\bfv_{il}[t] ~-~\bfv_{jl}[t]\parallel & \leq & (1-\gamma)\,\rho_l[t-1]\\&& \mbox{because the previous inequality holds for all $1\leq i,j\leq m$} \nonumber\\
\Rightarrow ~~~~
\rho_l[t] & \leq & (1-\gamma)\,\rho_l[t-1] \mbox{~~~~by (\ref{rho_equi})}
\label{e_rho_t_1} \\
\Rightarrow ~~~~
\Omega_l[t]-\mu_l[t] & \leq & (1-\gamma)\, (\Omega_l[t-1]-\mu_l[t-1])
\mbox{~~~~by definition of $\rho_l[t]$} \nonumber
\end{eqnarray}

 This proves (\ref{e_convergence}).

\section{Optimization of Step 2 of Asynchronous BVC}
\label{a_opt}

Property 1 of {\em Component \#1} of AAD described in Section
\ref{ss_async_suff}  is a consequence of
a stronger property satisfied by the AAD algorithm.

In AAD, each process $p_k$ sends out notifications to others
each time it adds a new tuple to its $B_k[t]$; the notifications
are sent over the FIFO links.
AAD defines a process $p_k$ to be a ``witness'' for process $p_i$
provided that (i) $p_k$ is known to have added at least $n-f$ tuples
to $B_k[t]$, and (ii) all the tuples that $p_k$ claims to have added
to $B_k[t]$ are also in $B_i[t]$.

AAD also ensures that each non-faulty process has at least
$n-f$ witnesses, ensuring that any two non-faulty processes
have at least $n-2f$ witnesses in common, where $n-2f\geq f+1$.
Thus, any two non-faulty processes $p_i$ and $p_j$ have at least one non-faulty
witness in common, say $p_k$. This, in turn, ensures (due to the manner in
which the advertisements above are sent) that
$B_i[t]\cap B_j[t]$ contains at least the first $n-f$ tuples
advertised by $p_k$.

Each process can keep track of the order in which the tuples advertised by
each process are received. Then, in Step 2 of the asynchronous approximate
BVC algorithm, instead of enumerating all the $n-f$-size subsets $C$ of $B_i[t]$,
it suffices to only consider those subsets of $B_i[t]$ that correspond
to the first $n-f$ tuples advertised by each witness of $p_i$.
Since there can be no more than $n$ witnesses, at most $n$ sets $C$ need
to be considered.
Thus, in this case $|Z_i|\leq n$.

Since each pair of non-faulty processes $p_i$ and $p_j$ shares a non-faulty 
witness, despite considering only $\leq n$ subsets in Step 2,
$Z_i$ and $Z_j$ computed by $p_i$ and $p_j$
contain at least one identical point, say, $\bfz_{ij}$. Our 
proof of correctness of the algorithm relied on the existence of such a point.

It should now be easy to see that the rest of the proof
of correctness will remain the same, with $\gamma$ being re-defined as
\[
\gamma ~=~\frac{1}{n^2}.
\]

\end{document}